\pdfoutput=1

\documentclass[conference,letterpaper]{IEEEtran}

\usepackage[utf8]{inputenc} 
\usepackage[T1]{fontenc}
\usepackage{url}
\usepackage{ifthen}
\usepackage{cite}
\usepackage[cmex10]{amsmath} 

\usepackage{cleveref}
\crefname{algocf}{alg.}{algs.}
\Crefname{algocf}{Algorithm}{Algorithms}
\usepackage{amssymb}
\usepackage{amsmath}
\usepackage{amsthm}
\usepackage{stmaryrd}
\usepackage{xcolor}
\usepackage[ruled, vlined]{algorithm2e}
\usepackage[makeroom]{cancel}

\usepackage{nicefrac}
\usepackage{enumitem}

\usepackage{tikz}
\usepackage{pgfplots}
\usepackage{tikzscale}
\pgfplotsset{compat=1.18}

\usetikzlibrary{pgfplots.groupplots}
\usepgfplotslibrary{fillbetween}

\definecolor{red}{HTML}{ca0020}
\definecolor{lightred}{HTML}{f4a582}
\definecolor{lightblue}{HTML}{92c5de}
\definecolor{green}{HTML}{008837}
\definecolor{blue}{HTML}{2c7bb6}

\pgfplotsset{
  tick label style={font=\small},
  every axis plot/.style={
  mark=none,
  line width=1pt, 
  grid=major},
}

\tikzset{every mark/.append style={scale=1.5}}
\tikzstyle{separator} = [thick,-, red, fill opacity=0.4]
\tikzstyle{pruned} = [red, fill opacity=0.1]

\newtheorem{theorem}{Theorem}[section]

\newtheorem{lemma}{Lemma}[section]
\theoremstyle{definition}
\newtheorem{definition}{Definition}[section]

\usepackage{amsmath,amsfonts,bm}

\def\eqref#1{equation~\ref{#1}}

\def\1{\bm{1}}

\def\rvu{{\mathbf{i}}}

\def\rvu{{\mathbf{u}}}

\def\rvx{{\mathbf{x}}}
\def\rvy{{\mathbf{y}}}
\def\rvz{{\mathbf{z}}}

\DeclareMathAlphabet{\mathsfit}{\encodingdefault}{\sfdefault}{m}{sl}
\SetMathAlphabet{\mathsfit}{bold}{\encodingdefault}{\sfdefault}{bx}{n}

\newcommand{\KL}{D_{\mathrm{KL}}}

\usepackage{mathtools}

\newcommand{\Nats}{\mathbb{N}}
\newcommand{\Ints}{\mathbb{Z}}
\newcommand{\Reals}{\mathbb{R}}
\newcommand{\Oh}{\mathcal{O}}

\newcommand{\Exp}{\mathbb{E}}
\newcommand{\Prob}{\mathbb{P}}

\newcommand{\Normal}{\mathcal{N}}
\newcommand{\Laplace}{\mathcal{L}}
\newcommand{\Unif}{\mathcal{U}}

\newcommand{\Ind}{\mathbf{1}}

\newcommand{\logtwo}{\log_2}
\newcommand{\exptwo}{\exp_2}

\newcommand{\lvlset}{\mathcal{H}}

\DeclarePairedDelimiterX{\infdivx}[2]{[}{]}{%
  #1\delimsize\| #2%
}
\DeclarePairedDelimiterX{\infdivxcolon}[2]{[}{]}{%
  #1\delimsize: #2%
}
\DeclarePairedDelimiterX{\infdivxarrow}[2]{(}{)}{%
  #1\delimsize\rightarrow #2%
}

\DeclarePairedDelimiterX{\infdivxtriple}[3]{[}{]}{%
  #1\delimsize: #2 \delimsize: #3%
}

\newcommand{\KLD}{\KL\infdivx}

\newcommand{\CSD}{D_{CS}\infdivx}
\newcommand{\phiD}{D^\phi\infdivx}
\newcommand{\ACSD}{D_{ACS}\infdivx}

\newcommand{\infD}{D_{\infty}\infdivx}
\newcommand{\MI}{\mathbb{I}\infdivxcolon}
\newcommand{\EFI}{\Psi\infdivxarrow}
\newcommand{\Ent}{\mathbb{H}}
\DeclarePairedDelimiter{\norm}{\lVert}{\rVert}
\DeclarePairedDelimiter{\abs}{\lvert}{\rvert}

\DeclarePairedDelimiterX{\innerProd}[2]{\langle}{\rangle}{%
    #1,#2%
}
\DeclareMathOperator*{\esssup}{ess\,sup}

\DeclareMathOperator\supp{supp}

\begin{document}

\title{On Channel Simulation with Causal Rejection Samplers} 

\author{%
  \IEEEauthorblockN{Daniel Goc}
  \IEEEauthorblockA{University of Cambridge\\
  Cambridge, UK\\
  dsg42@cam.ac.uk
  }
  \and
  \IEEEauthorblockN{Gergely Flamich}
  \IEEEauthorblockA{University of Cambridge\\
  Cambridge, UK\\
  gf332@cam.ac.uk
  }
}

\maketitle

\begin{abstract}
%
One-shot channel simulation has recently emerged as a promising alternative to quantization and entropy coding in machine-learning-based lossy data compression schemes.
However, while there are several potential applications of channel simulation -- lossy compression with realism constraints or differential privacy, to name a few -- little is known about its fundamental limitations.
In this paper, we restrict our attention to a subclass of channel simulation protocols called causal rejection samplers (CRS), establish new, tighter lower bounds on their expected runtime and codelength, and demonstrate the bounds' achievability.
Concretely, for an arbitrary CRS, let $Q$ and $P$ denote a target and proposal distribution supplied as input, and let $K$ be the number of samples examined by the algorithm. 
We show that the expected runtime $\Exp[K]$ of any CRS scales at least as $\exp_2(\infD{Q}{P})$, where $\infD{Q}{P}$ is the R{\'e}nyi $\infty$-divergence.
Regarding the codelength, we show that ${\KLD{Q}{P} \leq \CSD{Q}{P} \leq \Ent[K]}$, where $\CSD{Q}{P}$ is a new quantity we call the channel simulation divergence.
Furthermore, we prove that our new lower bound, unlike the $\KLD{Q}{P}$ lower bound, is achievable tightly, i.e.\ there is a CRS such that ${\Ent[K] \leq \CSD{Q}{P} + \logtwo (e + 1)}$.
Finally, we conduct numerical studies of the asymptotic scaling of the codelength of Gaussian and Laplace channel simulation algorithms.
\end{abstract}
\section{Introduction}
\noindent
Let $\rvx, \rvy \sim P_{\rvx, \rvy}$ be a pair of correlated random variables with $P_{\rvx, \rvy}$ shared between communicating parties Alice and Bob, who also share a source of common randomness.
Then, during one round of channel simulation, Alice receives ${\rvy \sim P_\rvy}$, and aims to send the minimum number of bits to Bob, so that he can simulate a single sample $\rvx \sim P_{\rvx \mid \rvy}$.
This procedure is a broad generalization of dithered quantization (DQ, \cite{ziv1985universal}) and provides a particularly natural alternative to quantization and entropy coding for machine-learning-based data compression.
\par
Channel simulation's advantage in machine learning stems from the fact that $P_{\rvx \mid \rvy}$ may be continuous (usually additive uniform, Gaussian or Laplace in practice), which opens the possibility to use powerful generative models for non-linear transform coding \cite{balle2020nonlinear}.
In just the last few years, various forms of channel simulation were used to compress images with variational autoencoders \cite{flamich2020compressing, agustsson2020universally, flamich2022fast,flamich2023faster} as well as implicit neural representations of various types of data \cite{guo2023compression, he2024recombiner}.
Furthermore, it has shown promise in differentally private federated learning \cite{shah2021optimal,shahmiri2023communication} and lossy compression with realism constraints \cite{Theis2021a}. 
\par
However, despite its growing popularity, characterising the fundamental limitations of channel simulation tightly remains an open problem.
Denoting the common randomness as $\Pi$, Li and Anantharam \cite{li2021unified} have shown that
\begin{align}
\MI{\rvx}{\rvy} \leq \Ent[\rvx \mid \Pi] \leq \MI{\rvx}{\rvy} + \logtwo(\MI{\rvx}{\rvy} + 1) + 4.732.
\nonumber
\end{align}
However, akin to entropy coding algorithms, there exist channel simulation algorithms with a much stronger one-shot guarantee.
Let us take greedy rejection sampling (GRS, presented in \Cref{alg:grs}) as our example, an algorithm which was first introduced by Harsha et al.\ \cite{harsha2007communication} for discrete distributions and was generalized in \cite{flamich2023adaptive} and \cite{flamich2023faster} to general probability spaces.
Given $\rvy \sim P_\rvy$, let us introduce the shorthand $Q \gets P_{\rvx \mid \rvy}$ and $P \gets P_\rvx$, and let $K$ be the code GRS outputs.
Then (see \cite{harsha2007communication}):
\begin{align}
\KLD{Q}{P} &\leq \Ent[K] \nonumber\\
&\leq \KLD{Q}{P} + 2\logtwo(\KLD{Q}{P} + 1) + \Oh(1).
\nonumber
\end{align}
Note, that sadly the asymptotics of the lower and upper bounds on the coding rate do not match; furthermore it can be shown that they cannot be tightened either \cite{li2018strong}.
In a similar vein, Agustsson and Theis \cite{agustsson2020universally} show that under the computational hardness assumption RP $\neq$ NP, the expected runtime of any ``general-purpose'' channel simulation algorithm must scale at least as $\exptwo(\KLD{Q}{P})$.
However, the expected runtime of all currently known algorithms scale as $\exptwo(\infD{Q}{P})$.
\par
\textit{Contributions.}
In this paper, we contribute to the understanding of the above mentioned fundamental limitations of channel simulation algorithms, by studying an important subclass of algorithms called causal rejection samplers (CRS).
In particular, we show that:
\begin{itemize}
\item The expected runtime of any CRS must scale at least as $\exptwo(\infD{Q}{P})$.
\item The expected one-shot coding rate of any causal rejection sampler that returns code $K$ is bounded below as 
\begin{align}
\KLD{Q}{P} \leq \CSD{Q}{P} \leq \Ent[K],
\nonumber
\end{align}
where $\CSD{Q}{P}$ is the \textit{channel simulation divergence}, a quantity that we define and analyze in \Cref{sec:channel_simulation_divergence}.
\item We prove that this lower bound is tight, and in particular it is achieved by GRS, with
\begin{align}
&\Ent[K] \leq \CSD{Q}{P} + \logtwo(e + 1) \nonumber \\
&\,\,\,\leq \KLD{Q}{P} + \logtwo(\KLD{Q}{P} + 1) + \logtwo(2e + 2).
\nonumber
\end{align}
\item We numerically investigate the asymptotic behaviour of Laplace and Gaussian channel simulation using CRS.  
\end{itemize}
\begin{algorithm}[t]
\SetAlgoLined
\DontPrintSemicolon
\SetKwInOut{Input}{Input}
\SetKwInOut{Output}{Output}
\Input{Target $Q$, Shared iid samples $X_1, X_2\hdots \sim P$.}
$L_1, S_1 \gets (0, 1)$ \;
\For{$k = 1$ \KwTo $\infty$}{
$U_k \sim \Unif(0, 1)$ \;
$\beta_k \gets \min\left\{1, \max\left\{0, \left(\frac{dQ}{dP}(X_k) - L_k \right) \Big / S_k\right\} \right\}$ \;
 \If{$U_k \leq \beta_k$}{
    \KwRet{$X_k, k$}
 }
 $L_{k + 1} \gets L_k + S_k$ \;
 $\lvlset_{k + 1} \gets \left\{ y \in \Omega \mid L_{k + 1} \leq \frac{dQ}{dP}(y) \right\}$ \;
 $S_{k + 1} \gets Q(\lvlset_{k + 1}) - L_{k + 1} \cdot P(\lvlset_{k + 1})$
}
\caption{Greedy rejection sampling, as formulated in \cite{flamich2023adaptive}.
See \Cref{sec:background} for notation.}
\label{alg:grs}
\end{algorithm}
\setlength{\textfloatsep}{14pt}
\section{Background}
\label{sec:background}
\noindent
\textit{Notation.} 
We denote the base-$2$ logarithm as $\logtwo$ and the natural logarithm as $\ln$. 
We denote the inverse function of $\logtwo$, the base-$2$ exponential, by $\exptwo$.
Let $Q$ and $P$ be Borel probability measures over a Polish space $\Omega$, with $Q \ll P$ and Radon-Nikodym derivative $dQ/dP$.
We define ${\KLD{Q}{P} = \Exp_{Z \sim Q}[\logtwo\nicefrac{dQ}{dP}(Z)]}$ and ${\infD{Q}{P} = \esssup_{x \in \Omega}\{\logtwo\nicefrac{dQ}{dP}(x)\}}$ where the essential supremum is taken with respect to $P$.
For a scalar $p$ we define the binary entropy $\Ent_b[p] = -p \logtwo(p) - (1-p) \logtwo(1-p)$.
For a discrete random variable $K$ with mass function $f$, we denote its Shannon entropy as $\Ent[K] = \Exp[-\logtwo f(K)]$.
Furthermore, we denote the $d$-dimensional Gaussian measure with mean $\mu$ and covariance $\Sigma$ as $\Normal(\mu, \Sigma)$, the $1$-dimensional Laplace measure with location $m$ and scale $s$ as $\Laplace(m, s)$ and the uniform measure on $(a, b)$ as $\Unif(a, b)$.
\par
\textit{One-shot Channel Simulation (OSCS).}
Let Alice and Bob be communicating parties, who share some source of common randomness, and let $\rvx, \rvy \sim P_{\rvx, \rvy}$ be a pair of correlated random variables.
Then, the OSCS problem is as follows: Alice receives a source symbol $\rvy \sim P_\rvy$, and her goal is to send the minimum number of bits to Bob so that he can simulate a sample $\rvx \sim P_{\rvx \mid \rvy}$.
Note, that in OSCS, Bob does not seek to learn $P_{\rvx \mid \rvy}$, he just needs to be able to simulate a sample from it.
An important special case is additive noise, i.e.\ when $\rvx, \rvy \in \Reals^d$ and $\rvx = \rvy + \epsilon$ for some random variable $\epsilon \perp \rvy$.
The most common case that arises in machine learning is additive Gaussian noise, i.e.\ $\epsilon \sim \Normal(0, I)$.
On the other hand, the additive uniform case, i.e.\ when $\epsilon \sim \Unif(-1/2, 1/2)$ can be efficiently solved by dithered quantization \cite{ziv1985universal, agustsson2020universally}.
While it can be extended to perform channel simulation with a richer family of distributions \cite{theis2021algorithms, hegazy2022randomized, flamich2023adaptive, ling2023vector}, dithered quantization relies on algebraic properties of the underlying space and thus cannot be extended to solve OSCS in general.
\par
\textit{Causal rejection samplers} (CRS) on the other hand, are a computational framework that solve the general OSCS problem given that Alice and Bob share an infinite sequence $\Pi = (X_1, X_2, \hdots)$ of i.i.d.\ $P_\rvx$-distributed random variables.
Intuitively, a CRS is a rejection sampler whose criterion for accepting a proposal sample depends only on samples it examined so far.
Following \cite{liu2018rejection}, we formally define CRS.
\begin{definition}[Causal Rejection Sampler]
\label{def:causal_rejection_sampling}
Let $(X_i)_{i \in \mathbb{N}}$ be a sequence of i.i.d.\ random variables with $X_i \sim P$.
Let $K$ be a finite stopping time adapted to the sequence and suppose that $X_K \sim Q$. We call such system a causal rejection sampler, and call $P$ the proposal and $Q$ the target distribution, respectively.
\end{definition}
Operationally, given $\rvy \sim P_{\rvy}$, Alice sets $Q \gets P_{\rvx \mid \rvy}$ as the target and $P \gets P_\rvx$ as the proposal, and uses the elements of $\Pi$ as samples proposed from $P_\rvx$ in her rejection sampler.
She then encodes the index $K$ of her accepted sample $X_K$ and transmits it to Bob.
Given $K$, Bob can recover Alice's accepted sample by selecting the $K$th sample from $\Pi$.
\par
\textit{Greedy Rejection Sampling} (GRS) is an example of a CRS and is presented in \Cref{alg:grs}; for a complete description and analysis we defer the reader to \cite{flamich2023adaptive} and \cite{flamich2023faster}.
In contrast to standard rejection sampling, GRS was designed with OSCS in mind, and hence it attempts to minimize $\Ent[K]$.
To this end, it accepts proposed samples as early as possible (``greedily''): 
for each element of the sample space $x \in \Omega$ there is a step $k$ such that GRS always accepts $x$ if it is proposed before step $k$, and always rejects $x$ if it is proposed after step $k$ \cite{harsha2007communication, flamich2023adaptive}.
\par
For the rest of our paper, we will be analyzing the one-shot performance of OSCS algorithms.
Hence, we will always implicitly assume that Alice has already received some fixed source symbol $\rvy \sim P_\rvy$, and will use the shorthand $Q = P_{\rvx \mid \rvy}$ and $P = P_\rvx$; 
average-case results can be obtained by taking expectations of the appropriate quantities with respect to $\rvy$.
\section{Runtime Bound for Causal Rejection Samplers} 
In this section, we tighten Agustsson and Theis' general runtime $\exptwo(\KLD{Q}{P})$ lower bound \cite{agustsson2020universally}, and show that any CRS's runtime must scale at least as $\exptwo(\infD{Q}{P})$.
\begin{theorem}
\label{thm:runtime_lower_bound}
Let $K$ be an index returned by a CRS with proposal $P$ and target $Q$. Then, 
\begin{align*}
\mathbb{E}[K] \ge \exptwo\left(\infD{Q}{P} \right).
\end{align*}
\end{theorem}
\begin{proof}
By definition of R{\'e}nyi $\infty$-divergence, there exists a sequence of sets $A_1, A_2, \dots$ such that $P(A_i) > 0$ for all $i$ and $Q(A_i) / P(A_i) \to \exptwo\left(\infD{Q}{P} \right)$ as $i \to \infty$. 
Now, fix some $i$. 
We now show that $\mathbb{E}[K] \ge Q(A_i) / P(A_i)$:
\begin{align*}
Q(A_i) = \mathbb{P}[X_{K} \in A_i] 
=&\: \sum_{k}{\mathbb{P}[K = k, X_k \in A_i]}
\\
\le&\: \sum_{k}{\mathbb{P}[K \ge k, X_k \in A_i]}
\\
\stackrel{(a)}{=}&\: \sum_{k}{\mathbb{P}[K \ge k] \cdot P(A_i)}
\\
=&\: P(A_i) \cdot \mathbb{E}[K]
\end{align*}
where the equality (a) follows as the event $K \ge k$ is independent of variable $X_k$, since $K$ is a stopping time. 
\end{proof}
In Appendix~\ref{appendix:extended_results}, we generalize \Cref{thm:runtime_lower_bound} in two ways. 
We show that it holds for an extended definition of CRS, as well as for a certain non-causal sampling framework used by e.g.\ the Poisson functional representation / A* coding \cite{li2018strong,flamich2022fast}.
This strengthens a result of Liu and Verd{\'u} \cite[cor.\ 4]{liu2018rejection}, who prove an asymptotic version of our \cref{thm:runtime_lower_bound}.
\section{The Channel Simulation Divergence}
\label{sec:channel_simulation_divergence}
\noindent
In this section, we define a statistical distance between two distributions, which we call the \textit{channel simulation divergence}.
In \Cref{sec:codelength}, we use it to bound the codelength of causal rejection samplers.
Interestingly, it is also connected to 
Li and El Gamal's \textit{excess functional information} \cite{li2018strong} and Hegazy and Li's \textit{layered entropy} \cite{hegazy2022randomized};
we discuss this in Appendix~\ref{appendix:csd_layered_entropy_connection}.
\begin{definition}[Channel Simulation Divergence]
\label{def:channel_simulation_divergence}
Let $Q$ and $P$ be Borel probability measures over a Polish space $\Omega$ with $Q \ll P$.
Then, let
\begin{align}
w(h) = \mathbb{P}_{X \sim P}\left[\frac{dQ}{dP}(X) \ge h \right] = P\left(\frac{dQ}{dP} \ge h \right)
\nonumber
\end{align}
be the $P$-width of the $h$-superlevel set of $dQ/dP$.
Then, the channel simulation divergence of $Q$ from $P$ is
\begin{align}
\CSD{Q}{P} = -\int_{h = 0}^{\infty}w(h) \logtwo w(h) \, dh.
\label{eq:channel_simulation_divergence_def}
\end{align}
\end{definition}
The following lemma collects some of the most important properties of the divergence.
\begin{lemma}
\label{lemma:csd_properties}
Let $Q$ and $P$ be as in \Cref{def:channel_simulation_divergence}.
Then:
\begin{enumerate}[
label=(\Roman*),
]
\item{$\CSD{Q}{P} \geq 0$.
\hfill
\textbf{(Non-negativity)}} 
\item{$\CSD{Q}{P} = 0 \Leftrightarrow Q = P$.
\hfill
\textbf{(Positivity)}} 
\item{
$\CSD{Q}{P}$ is convex in $Q$ for fixed $P$.
\hfill
\textbf{(Convexity)}}
\item{
For $\kappa = \KLD{Q}{P}$, we have \hfill
\textbf{(KL bound)}}
\begin{align}
\kappa \leq \CSD{Q}{P} \leq \kappa + \logtwo(\kappa + 1) + 1. \nonumber
\end{align}
\item When $\Omega$ is discrete, for $x \in \Omega$ and setting $Q = \delta_x$ we have
${\CSD{Q}{P} = -\logtwo P(x)}$. \hfill
\textbf{(Self-information)}
\end{enumerate}
\end{lemma}
\begin{proof}
(I), (II):
First, observe that by definition ${0 \leq w(h) \leq 1}$ for all $h \geq 0$, with $w(0) = 1$ and decreasing in $h$.
Note also that $\int_0^\infty w(h)\,dh = \Exp_{X \sim P}[\nicefrac{dQ}{dP}(X)] = 1$, where the first equality follows from the Darth Vader rule \cite{muldowney2012darth}.
Hence, $w(h)$ is the probability density of a non-negative random variable $H$.
As $w(h) \leq 1$, we find ${\CSD{Q}{P} = \Exp_H[-\logtwo w(H)] \geq 0}$, with equality if and only if $w$ is identically $1$ on its entire support, implying ${w(h) = \Ind[0 \leq h \leq 1]}$.
However, ${w(1) = 1 \Leftrightarrow \nicefrac{dQ}{dP} \stackrel{a.e.}{=}1 \Leftrightarrow Q = P}$.
\vspace{0.2cm}
\par
(III):
%
%
We begin our proof by first showing the following integral representation of the channel simulation divergence:
\begin{align}
\CSD{Q}{P}\!\cdot\! \ln 2 &= -1 + \int_0^1 \int_0^\infty \frac{\min\{w(h), y\}}{y} \,dh \, dy.
\label{eq:channel_simulation_divergence_integral_rep}
\end{align}
To this end, let $\phi(x) \!=\! -x \ln x$ and ${\omega(y) \!=\! \int_0^\infty \Ind[w(h)\! \geq\! y]\, dh}$.
By \cref{eq:channel_simulation_divergence_def}, we also have ${\CSD{Q}{P}\cdot \ln 2 = -\int_0^1 \phi(y)\, d\omega(y)}$, which we integrate by parts twice to obtain \cref{eq:channel_simulation_divergence_integral_rep}.
%
%
\par
Fix distributions $P, Q_1, Q_2$ satisfying $Q_1, Q_2 \ll P$, some $\lambda \in [0, 1]$ and set $Q = \lambda Q_1 + (1 - \lambda) Q_2$.
Let $r_i = \nicefrac{dQ_i}{dP}$,
with corresponding $P$-width functions $w, w_1, w_2$.
Now, define $\mu_y(w) = \int_0^\infty \min\{w(h), y\} \, dh$.
Then, by \cref{eq:channel_simulation_divergence_integral_rep} it suffices to show that for all $y \in (0, 1)$:
\vspace{-0.1cm}
\begin{align}
\label{eq:convexity1}
\mu_y(w) \leq \lambda\mu_y(w_1) + (1 - \lambda)\mu_y(w_2).
\end{align}
%
%
%
Now fix $y \in (0, 1)$. 
Without loss of generality assume that $P$ is non-atomic.
For each $i$ we can find super-level set $A_i$ of $r_i$ with $P$-width $y$:
\vspace{-0.1cm}
\begin{align*}
    P(A_i) = y \ \ \ \ \ x_1 \in A_i, \, x_2 \in A_i^{c} \ \Rightarrow \ r_i(x_1) \ge r_i(x_2)
\end{align*}
One can verify that:
\vspace{-0.2cm}
\begin{align}
\label{eq:level_set}
\mu_y(w_i) = 1 -\int_{A_i^{c}} r_i \,dP.
\end{align}
Set $l_i = \inf_{A_i} r_i(x)$. For each $i = 1, 2$ we transform:
\vspace{-0.1cm}
\begin{align*}
\int_{A_i^c} r_i \,dP
&= \int_{A_i^c \cap A^c} r_i \,dP + \int_{A_i^c \cap A} r_i \,dP
\\
&\le \int_{A_i^c \cap A^c} r_i \,dP + l_i \cdot P(A_i^c \cap A)
\\
&\stackrel{(a)}{=} \int_{A_i^c \cap A^c} r_i \,dP + l_i \cdot P(A_i \cap A^c)
\\
&\le \int_{A_i^c \cap A^c} r_i \,dP + \int_{A_i \cap A^c} r_i \,dP
\\
&= \int_{A^c} r_i \,dP
\end{align*}
where equality (a) follows from $P(A_i) = P(A)$. Finally, noting that $r(x) = \lambda r_1(x) + (1-\lambda) r_2(x)$, the above in combination with \cref{eq:level_set} gives \cref{eq:convexity1}.
\vspace{0.0cm}
\\
\textit{Remark.} Claims (I -- III) can be shown for a broad generalization of $\CSD{Q}{P}$; we provide the details in Appendix~\ref{appendix:csd_convexity}.
\vspace{-0.3cm}
\par
(IV): Before we show the bound, we first note that 
\begin{align}
\KLD{Q}{P} &= \int_\Omega \frac{dQ}{dP} \logtwo \frac{dQ}{dP}\, dP \nonumber\\
&= -\int_0^\infty h \logtwo h \, dw(h) \nonumber\\
&= \logtwo e + \int_0^\infty w(h) \logtwo h\,dh,
\label{eq:kl_width_identity}
\end{align}
where the second equality follows by jointly integrating over points where $dQ/dP = h$ and the third equality follows from integration by parts. 
\par
We will now focus on showing the lower bound. 
First, define $W(h) = \int_0^h w(\eta)\,d\eta = \Prob[H \leq h]$.
Note, that since $w(h)$ is decreasing, ${W(h) \geq h \cdot w(h)}$, thus $\logtwo h \leq -\logtwo w(h) + \logtwo W(h)$.
Substituting this into \cref{eq:kl_width_identity} and noting that $\int_0^\infty w(h) \logtwo W(h) \,dh = \int_0^1 \logtwo u \, du = -\logtwo e$, we obtain the desired lower bound.
\par
Next, we show the upper bound.
We begin by noting that $-\logtwo e = \int_0^\infty w(h) \logtwo (1 - W(h))\, dh$.
Then,
\vspace{-0.1cm}
\begin{align}
\CSD{Q}{P} - \KLD{Q}{P}
= \int_{h = 0}^{\infty} w(h) \logtwo\left(\frac{1 - W(h)}{h \cdot w(h)} \right) \,dh
\nonumber
\end{align}
We now split the integral into two integrals, one with domain of integration $[0, 1)$ representing the behavior of $w$ when $dQ/dP < 1$ and a second with $[1, \infty)$, representing $dQ/dP \geq 1$.  
We now bound the two integrals separately.
\par
We observe, that the integral on $[1, \infty)$ can be seen as a scaled expectation over the random variable $H$ truncated to $[1, \infty)$.
Now, setting $T = \Prob[H \geq 1] = 1 - W(1)$, applying Jensen's inequality, and integrating by parts, we get
\begin{align}
T\Exp_{H \mid H \geq 1}&\left[\logtwo \frac{1 - W(H)}{H\cdot w(H)}\right] \leq T \logtwo \Exp_{H \mid H \geq 1}\left[\frac{1 - W(H)}{H\cdot w(H)}\right] \nonumber\\
&\,\,\,= T \logtwo\left(\int_1^\infty w(h)\logtwo h\, dh\right) - T\logtwo T
\nonumber
\\
& \stackrel{\text{\cref{eq:kl_width_identity}}}{\le} T \logtwo\KLD{Q}{P} - T\logtwo T.
\label{eq:upper_csd_kld_integral_bound}
\end{align}
Now, we focus on bounding the first integral, over $[0, 1)$.
Let $\mathcal{W}\! =\! \{\omega\! \in\! \mathcal{L}_1([0, 1])\mid 0\! \leq\! \omega\! \leq 1, \norm{\omega}_1\! =\! S, w' \leq 0\}$, 
where $S\! =\! W(1)$ and, as a slight abuse of notation, $\mathcal{L}_1([0, 1])$ denotes the Lebesgue integrable functions on $[0, 1]$.
Then:
\begin{align}
\int_0^1\!\! w(h)\!\ln \frac{1}{h} \, dh\! &\leq \!\sup_{\omega \in \mathcal{W}} \int_0^1 \!\!\omega(h)\ln \frac{1}{h} \, dh
\!=\! S(1\! -\! \ln S).
\label{eq:partial_kl_integral_lower_bound}
\end{align}
Also, note that for $U \sim \Unif(0, 1)$ by Jensen's inequality,
\begin{align}
-\int_{0}^{1}\!\! w(h) \ln w(h) \,dh = \Exp[-w(U)\ln w(U)] \le -S \ln S.
\label{eq:xlogx_integral_bound_jensen}
\end{align}
Finally, combining $\int_0^1 w(h) \ln(1 - W(h))\,dh = -T \ln(T) - S$ with \cref{eq:partial_kl_integral_lower_bound,eq:xlogx_integral_bound_jensen}, we get
\begin{align}
\int_0^1 w(h) \logtwo \frac{1 - W(h)}{h \cdot w(h)}\, dh \leq -T \logtwo T - 2 S \logtwo S.
\label{eq:difference_integral_bound_on_0_1}
\end{align}
Putting the upper bounds in \cref{eq:upper_csd_kld_integral_bound,eq:difference_integral_bound_on_0_1} together, we get
\begin{align}
\CSD{Q}{P} \leq \KLD{Q}{P} + T \logtwo\KLD{Q}{P} + 2\Ent_b[T].
\nonumber
\end{align}
Finally, for $0 \leq t \leq 1$ and $a > 0$ we apply the inequality $t \logtwo(a) + 2\Ent_b[t] \leq \logtwo (a + 1) + 1$, which finishes the proof.
We note that the above upper bound can be improved using variational methods.
\vspace{0.1cm}
\par
(V): follows from direct calculation.
\end{proof}
\vspace{-0.1cm}
\noindent
In Appendix~\ref{appendix:alternative_csd}, we provide some additional results for $\CSD{Q}{P}$ as well as define and analyze a related divergence.
\section{The Codelength of Causal Rejection Samplers}
\label{sec:codelength}
\noindent
We are now ready to present our two main results on CRSs.
First, we show that the codelength of a CRS is always bounded from below by the channel simulation divergence.
\begin{theorem}
\label{thm:crs_index_entropy_lower_bound}
Let $K$ be the index returned by any CRS. 
Then:
\begin{align}
\CSD{Q}{P} \leq \Ent[K] \nonumber
\end{align}
\end{theorem}
\begin{proof}
The key step in the proof is the following lemma:
\begin{lemma}
\label{lemma:index_entropy_csd_difference}
For some $k \in \Ints^+$, let $Q_k$ be the distribution of $X_K | K \ge k$. Then:
\vspace{-0.1cm}
\begin{align*}
\mathbb{P}[K &\ge k] \left(\Ent[K | K \ge k] - \CSD{Q_k}{P}\right)
\\
&\quad \ge \mathbb{P}[K \ge k + 1] \left(\Ent[K | K \ge k + 1] - \CSD{Q_{k + 1}}{P} \right)
\end{align*}
\end{lemma}
\begin{proof}
Set $q = \mathbb{P}(K = k | K \ge k) = 1 - p$. Expanding the Shannon entropy, we get:
\begin{align}
    \label{Shannon_expanded}
    \Ent[K | K \ge k] = \Ent_b[q] + p \Ent[K | K \ge k + 1]
\end{align}
Let $\widetilde{Q} \sim X_K | K = k$.
Then $Q_k = p Q_{k + 1} + q \widetilde{Q}$, thus by the convexity of the channel simulation divergence we obtain
\begin{align}
    \label{DCS_convex}
    \CSD{Q_k}{P} \le p \CSD{Q_{k + 1}}{P} + q \CSD{\widetilde{Q}}{P}
\end{align}
As $K$ is a stopping time, we find $P \sim X_k | K \ge k$. It follows that $\nicefrac{d\widetilde{Q}}{dP} \le q^{-1}$ holds $P$ - almost everywhere. Thus defining $\widetilde{w}(h) = P\left(\nicefrac{d\widetilde{Q}}{dP} \ge h \right)$ we find $\widetilde{w}(h) = 0$ for all $h > q^{-1}$. Setting $\phi(x) = -x \logtwo(x)$ by Jensen's inequality:
\vspace{-0.2cm}
\begin{align*}
\CSD{\widetilde{Q}}{P} &= \int_{0}^{q^{-1}}{\phi(\widetilde{w}(h)) \,dh}
\\
&\le q^{-1} \phi\left(q \int_{0}^{q^{-1}} \widetilde{w}(h) \,dh\right)
\\
&= -\logtwo(q)
\end{align*}
Combining the above with \cref{Shannon_expanded,DCS_convex} and $\Ent_b[q] \ge -q \logtwo(q)$ gives the thesis.
\end{proof}
\vspace{-0.1cm}
Iterating the \Cref{lemma:index_entropy_csd_difference} allows us to prove the theorem for any CRS with uniformly bounded stopping time $K$. 
The idea for the remaining part of the proof is to take terminated stopping times.
Let $\min\{K, L\}$ be the $L$-th terminated stopping time. Set $p_L = \mathbb{P}[K \le L]$. One can check that in limit:
\begin{align}
\label{eq:Shannon_limit}
\Ent[\min\{K, L\}] \rightarrow \Ent[K] \quad\text{as}\quad L \rightarrow \infty
\end{align}
Let $X_{\min\{K, L\}} \sim Q_L$. 
Since $\mathbb{P}[\min\{K, L\} \ge L + 1] = 0$, using \Cref{lemma:index_entropy_csd_difference} we find:
\vspace{-0.1cm}
\begin{align*}
\Ent[\min\{K, L\}] \ge \CSD{Q_L}{P}
\end{align*}
Finally, in Appendix~\ref{appendix:causal_rejection_samplers}, we show
\vspace{-0.1cm}
\begin{align}
\label{eq:last_limit}
\limsup_{L \rightarrow \infty} \CSD{Q_L}{P} \ge \CSD{Q}{P},
\end{align}
which, taken together with \cref{eq:Shannon_limit} finishes the proof.
\end{proof}
\vspace{-0.1cm}
As our second main result, we show that the lower bound in \Cref{thm:crs_index_entropy_lower_bound} is achieved by \Cref{alg:grs}.
\begin{theorem}
\label{thm:grs_achievability}
Let $K$ be the index \Cref{alg:grs} returns. 
Then:
\begin{align*}
\Ent[K] \leq \CSD{Q}{P} + \logtwo(e + 1).
\end{align*}
\end{theorem}
\begin{proof}
Set $w(h) = P\left(\nicefrac{dQ}{dP} \ge h \right)$. 
Let $S_k$ and $L_k$ be as defined by \Cref{alg:grs}.
Let ${q_k = \int_{0}^{1}{w(S_k h + L_k) \,dh}}$.
One can check that ${L_{k + 1} = L_k + S_k}$ and ${S_{k + 1} = S_k \cdot (1 - q_k)}$.
By Lemma III.1 of \cite{flamich2023adaptive}, $S_k = \Prob[K \geq k]$
and hence we have ${p_k = \mathbb{P}[K = k] = S_k \cdot q_k}$.
%
The motivation behind the above notation is that it allows us to express the $P$-width function of ${Q_k \sim X_K \mid K \ge k}$ as ${P(\nicefrac{dQ_k}{dP} \ge h) = w(S_k h + L_k)}$, which can be shown via direct calculation.
\par
Set $\phi(x) = -x\logtwo(x)$. The key step is to show for $k \ge 0$:
\begin{align}
\label{ineq:GRS_main_step}
\begin{split}
\Ent[K] \le &\logtwo(e + 1) + \int_{0}^{L_k}{\phi(w(h)) \, dh}
\\
&+ S_k \left(\Ent[K | K \ge k] - \logtwo(e + w(L_k))\right)
\end{split}
\end{align}%
We proceed by induction. 
The base case $k = 1$ is trivial. 
Now fix some $k \ge 1$. 
For the inductive step, we need to prove:
\begin{align}
\: S_k &\left(\Ent[K | K \ge k] - \logtwo(e + w(L_k))\right) - \int_{L_k}^{L_{k+1}}{\phi(w(h)) \,dh}
\nonumber \\
&\,\,\,\le S_{k+1} \left(\Ent[K | K \ge k+1] - \logtwo(e + w(L_{k+1}))\right)
\label{ineq:GRS_inductive_step}
\end{align}
Expanding the Shannon entropy:
\vspace{-0.1cm}
\begin{align}
    \label{eq:GRS_expand_entropy}
    \Ent[K | K \ge k] = \Ent_b[q_k] + (1 - q_k) \Ent[K | K \ge k + 1]
\end{align}
Furthermore, we note that:
\begin{align}
    \label{eq:GRS_expand_integral}
    \int_{L_k}^{L_{k+1}}{\phi(w(h)) \,dh} = S_k \int_{0}^{1}\phi(w(S_k h + L_k)) \,dh
\end{align}
Utilizing \cref{eq:GRS_expand_entropy,eq:GRS_expand_integral}, we transform \cref{ineq:GRS_inductive_step} into:
\begin{align}
\label{eq:GRS_transformed_inductive_step}
\Ent[q_k] &- \logtwo(e + w(L_k))
\\
\le&\: \int_{0}^{1}{\phi(w(S_k h + L_k)) \,dh} - (1-q_k) \logtwo(e + w(L_{k + 1}))
\nonumber
\end{align}
Set $f(x) = w(S_k x + L_k)$. Let $m = f(1)$ and $M = f(0)$.
Note, that since $w$ is decreasing, $m$ and $M$ are the minimum and the maximum of $f$ on $[0, 1]$, respectively. 
Now, setting $I = \int_{0}^{1} f(x) \,dh = q_k$, ineq. (\ref{eq:GRS_transformed_inductive_step}) becomes:
\begin{align}
\Ent_b[I]\! -\! \logtwo(e + M) 
\!\le\! \int_{0}^{1}\!\!{\phi(f(x)) \,dx}\! -\! (1 - I) \logtwo(e + m),
\nonumber
\end{align}
which we show to hold in Appendix~\ref{appendix:causal_rejection_samplers}.
\par
For the rest of the proof we assume that $\Ent[K] < \infty$. Having shown ineq. (\ref{ineq:GRS_main_step}), our final step is to take limit as $k \rightarrow \infty$. Note that $S_k \rightarrow 0$, thus $S_k \logtwo(e + w(L_k)) \rightarrow 0$. 
Furthermore:
\begin{align*}
    S_k \Ent[K | K \ge k] = -\phi(S_k) + \sum_{i = k}^{\infty} \phi(p_i)
\end{align*}
Since $\sum_{i = 1}^{\infty} \phi(p_k) = \Ent[K] < \infty$, the above tends to zero. 
Now, taking $k \rightarrow \infty$ in ineq. (\ref{ineq:GRS_main_step}) finishes the proof.
\end{proof}
We provide some additional related results in Appendix~\ref{appendix:grs_additional_results}.
\section{Laplace and Gaussian Channel Simulation}
\noindent
This section briefly investigates the behaviour of $\CSD{Q}{P}$ when $Q$ and $P$ are both Laplace or Gaussian distributions.
In particular, note that we can obtain a joint bound by combining part (IV) of \Cref{lemma:csd_properties} with \Cref{thm:crs_index_entropy_lower_bound,thm:grs_achievability}: 
\begin{align}
&\KLD{Q}{P}\! \leq\! \CSD{Q}{P}\! \leq\! \Ent[K]\! \leq\! \CSD{Q}{P} + \logtwo(e + 1) \nonumber\\
&\,\,\, \leq\! \KLD{Q}{P}\! +\! \logtwo(\KLD{Q}{P} \!+\! 1)\! +\! \logtwo(e\! +\! 1)\! +\! 1
\label{eq:joint_sandwich_bound}
\end{align}
Hence, by computing $\CSD{Q}{P}$, not only do we get a very precise estimate on the optimal performance of CRS for channel simulation, but we can also gauge how tight the $\KLD{Q}{P}$-based bound is as well.
To this end, let us define ${\Delta(Q, P) = \CSD{Q}{P} - \KLD{Q}{P}}$, which is the main quantity we plot in \cref{fig:numerical_results}.
\par
\textit{One-dimensional Laplace channel simulation.}
Let $0 \!< \!b \!\leq\! 1$, and set $Q = \Laplace(0, b)$ and $P = \Laplace(0, 1)$.
Then, using the formula for $w$ in this case (Appendix G.5 of \cite{flamich2023gprs}), we find that
\begin{align}
\CSD{Q}{P} \cdot \ln 2 = b + \psi(1 / b) + \gamma - 1,
\nonumber
\end{align}
where $\psi$ is the digamma function and $\gamma$ is the Euler-Mascheroni constant.
Using $-x \leq \psi(1/x) + \ln x \leq -x / 2$ for $x > 0$ from eq.\ 2.2 in \cite{alzer1997some}, this shows that
\begin{align}
\gamma - b \leq \Delta(Q, P) \cdot \ln 2 \leq \gamma - b / 2.
\nonumber
\end{align}
We plot the exact difference in \cref{fig:numerical_results} (A), which together with the above demonstrates that in this case, the $\KLD{Q}{P}$-based upper bound in \cref{eq:joint_sandwich_bound} can be tightened to $\KLD{Q}{P} + \Oh(1)$.
\par
\textit{Asymptotic Gaussian channel simulation.}
Next, we investigate the case when $Q = \Normal(\mu, \sigma^2)^{\otimes d}$ and $P = \Normal(0, 1)^{\otimes d}$ as $d \to \infty$, when $0 < \sigma < 1$.
Sadly, there does not appear to be an analytic formula for $\CSD{Q}{P}$, even when $\mu = 0$.
However, as \cref{eq:channel_simulation_divergence_def} is an integral over $\Reals^+$, it lends itself nicely to numerical quadrature.
Surprisingly, as \cref{fig:numerical_results} (B) shows,
\begin{align}
\Delta(Q, P) \approx \nicefrac{1}{2} \cdot \logtwo(\KLD{Q}{P} + 1),
\label{eq:conjectured_gaussian_kl_scaling}
\end{align}
which shows that neither the lower nor the upper bound based on $\KLD{Q}{P}$ in \cref{eq:joint_sandwich_bound} is tight!
\begin{figure}[t]
\centering
\ref{legend:experiments}
\\
\includegraphics{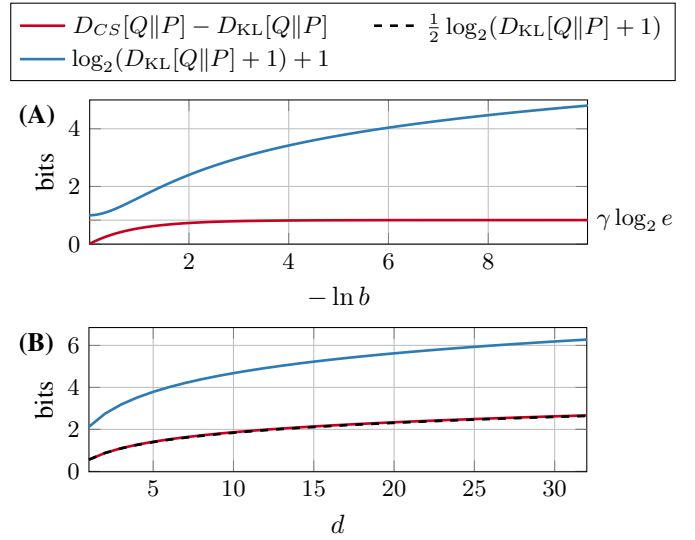}
\vspace{-0.6cm}
\caption{Numerical demonstration of the looseness of the channel simulation bound in \cref{eq:joint_sandwich_bound}.
\textbf{(A)} We plot $\Delta(Q, P)$ for $Q = \Laplace(0, b)$ and $P = \Laplace(0, 1)$ as a function of the target log-precision $-\ln b$.
\textbf{(B)} We plot $\Delta(Q, P)$ for $Q = \Normal(1, 1/4)^{\otimes d}$ and $P = \Normal(0, 1)^{\otimes d}$ as a function of the dimension $d$.
}
\label{fig:numerical_results}
\vspace{-0.2cm}
\end{figure}
\section{Discussion and Future Work}
\noindent
In this paper, we analyzed the expected runtime and codelength of an important class of channel simulation algorithms, called causal rejection samples.
In particular, we showed that for a target distribution $Q$ and proposal $P$, the expected runtime of any causal rejection sample must scale at least as $\exp(\infD{Q}{P})$.
Furthermore, we showed that their expected codelength is bounded below by the channel simulation divergence $\CSD{Q}{P}$, a quantity we introduced and analyzed in the paper.
We also showed that this lower bound is achievable within constant number of bits.
Finally, we gave precise numerical estimates of the achievable codelength for certain Laplace and Gaussian channel simulation problems, which suggested surprising $\KLD{Q}{P}$-based scaling too.
\par
Our work suggests several directions for future exploration:
First, we conjecture that  \Cref{thm:runtime_lower_bound,thm:crs_index_entropy_lower_bound} should hold for Harsha et al.'s general, non-causal sampling framework \cite{harsha2007communication}.
Second, we ask if $\CSD{Q}{P}$ can be used to give lower bounds for the $\alpha$-fractional moments of the index returned by a CRS for $\alpha \in (0, 1)$, akin to Theorem 3 of \cite{liu2018rejection}. 
Finally, we leave a rigorous proof of \cref{eq:conjectured_gaussian_kl_scaling} for future work.

\bibliographystyle{IEEEtran}
\bibliography{references}

\clearpage
\appendices

Throughout the appendices, we always assume that $P, Q$ are the proposal and target distributions respectively, with $Q \ll P$.

\section{Expanded Definitions}
\label{appendix:extended_results}

First we generalize the notion of CRS.
\begin{definition}[Expanded Causal Rejection Sampler]
    Let $I$ be a countable set and $\{X_i: i \in I\}$ a collection of random variables. Let $N_1, N_2, ... , N_K$ be a random finite sequence taking pairwise different values from set $I$, with the property that $X_{N_K} \sim Q$. Furthermore assume that for any $k$ and sequence $n_1, \dots , n_k$, conditional on $K \ge k$ and $N_1 = n_1, \dots , N_K = n_k$, the distribution of $X_{N_k}$ is $P$. We call such system an expanded causal rejection sampler with proposal $P$ and target $Q$.
\end{definition}

The interpretation of such defined Causal Rejection Sampler is as follows: There is a player with some strategy and a set of variables $\{X_i: i \in I\}$. At any point in time, the player chooses one new index $N_k$ and looks up a variable $X_{N_k}$, the distribution of which follows $P$. Once the player reads the value $X_{N_k}$, he can either choose to output the index $N_k$ - thus setting $K = k$ and terminating sequence $N_1, ... , N_K$ - or he can continue on to sample new variables.

The revitalized definition of Causal Rejection Sampler has 2 key advantages: firstly, it generalizes the traditional notion of Causal Rejection Sampler defined using stopping time, as introduced in \cite{liu2018rejection}. Indeed, we have:

\begin{theorem}
    Let $\{X_i: i \in \mathbb{N} \}$ be i.i.d. variables with distribution $P$. Let $K$ be a finite stopping time adapted to $X_1, X_2, \dots$. Then $N_1 = 1, N_2 = 2, \dots , N_K = K$ is a Causal Rejection Sampler.
\end{theorem}
\begin{proof}
    By assumption $K$ is always finite. Pick any sequence $n_1, n_2, \dots , n_k \in \mathbb{N}$. If $n_j \ne j$ for some $j$, then $\mathbb{P}[K \ge k, N_1 = n_1, \dots , N_K = n_k] = 0$ and the requirement on conditional distribution of $X_{N_k}$ is vacuously satisfied. Hence from now on we assume that $n_j = j$ for all $j$.

    Consider the event $E = \{K \ge k \} = \{K \le k - 1\}^{c}$. Since $K$ is a stopping time, $E$ belongs to $\sigma(X_1, \dots, X_{k-1})$. Therefore $E$ is independent of $X_k$. Thus:
    \begin{align*}
        X_k | \{K \ge k, N_1 = n_1, \dots , N_k = n_k\} \sim X_k | E \sim X_k \sim P
    \end{align*}
    as desired.
\end{proof}
The second advantage of our revitalised definition of Causal Rejection Sampler is the fact, that it does not require an ordering of indices - much like Shannon entropy.
\par
Now, we define a computational  model for non-causal sampling, and show that \Cref{thm:runtime_lower_bound} still holds in this extended framework.
Notably, this non-causal computational model covers the setting for the Poisson functional representation (PFR) / A* coding \cite{li2018strong,flamich2022fast}.
It is an interesting direction for future work to investigate the ``gap in computational power'' between our following definition, and Harsha et al.'s general non-causal sampling framework \cite{harsha2007communication,liu2018rejection}.
\begin{definition}[A*-like non-causal sampler]
\label{def:astar_like_sampler}
Let $(X_i)_{i \in \Nats}$ be a sequence of i.i.d.\ random variables with $X_i \sim P$, and let $K$ be a stopping time adapted to the sequence.
Now, let $1 \leq N \leq K$ be a random variable, such that $X_{N} \sim Q$.
We call such a system an A*-like non-causal sampler, in analogy to the sampling rule used by A* coding.
\end{definition}
Intuitively, an A*-like non-causal sampler only needs to check the first $K$ samples of the sequence $(X_i)$, after which it can select one of the examined samples such that its distribution will be guaranteed to follow the target $Q$.
Now note, that $N$ may depend on the entire sequence $(X_i)_{i \in \Nats}$.
We are now ready to extend \Cref{thm:runtime_lower_bound}.
\begin{theorem}
Let $P$, $Q$, $(X_i)_{i \in \Nats}$, $K$ and $N$ be as in \Cref{def:astar_like_sampler}.
Then,
\begin{align}
\exp(\infD{Q}{P}) \leq \Exp[K].
\end{align}
\end{theorem}
\begin{proof}
The statement of the theorem holds vacuously when $\Prob[K < \infty] < 1$.
Hence, for the rest of the proof, assume that $\Prob[K < \infty] = 1$.
\par
By definition of R{\'e}nyi $\infty$-divergence, there exists a sequence of sets $A_1, A_2, \dots$ such that $P(A_i) > 0$ for all $i$ and $Q(A_i) / P(A_i) \to \exp\left(\infD{Q}{P} \right)$ as $i \to \infty$. 
Now, fix some $i$. 
We will show that $\mathbb{E}[K] \ge Q(A_i) / P(A_i)$, thus finishing the proof.
Thus,
\begin{align}
&Q(A_i) = \Prob[X_{N} \in A_i] = \sum_{n = 1}^\infty \Prob[N = n, X_{n} \in A_i]
\nonumber \\
&= \sum_{n = 1}^\infty \Prob[K \geq n, N = n, X_{n} \in A_i] \\
&\,\,\,\,\stackrel{\text{(a)}}{=}\sum_{n = 1}^\infty \Prob[K \geq n]P(A_i)\Prob[N = n \mid K\geq n, X_n \in A_i] \nonumber \\
&\,\,\,\,\leq P(A_i)\sum_{n = 1}^\infty \Prob[K \geq n] \nonumber\\
&\,\,\,\,= P(A_i) \Exp[K].\nonumber
\end{align}
Eq.\ (a) follows from ${\Prob[X_n \in A_i \mid K \geq n] = P(A_i)}$, since $X_n \perp K$ for $K \geq n$ as $K$ is a stopping time.
Dividing both sides of the inequality by $P(A_i)$ finishes the proof.
\end{proof}
\section{The Connection Between the Channel Simulation Divergence, Excess Functional Information and Layered Entropy}
\label{appendix:csd_layered_entropy_connection}
\noindent
\textbf{Connection with the excess functional information:}
For two correlated random variables $\rvx, \rvy \sim P_{\rvx, \rvy}$, 
Li and El Gamal define the excess functional information as \cite{li2018strong}:
\begin{align}
\EFI{\rvx}{\rvy} = \inf_{\substack{\rvz\,:\, \rvz \perp \rvx \\ \Ent[\rvy \mid \rvx,\, \rvz]\, =\, 0}} \!\!\! \MI{\rvx}{\rvz \mid \rvy}
\end{align}
Then, letting $\phi(x) = -x \logtwo x$, they show that for discrete $\rvy$ with posterior $p(y \mid \rvx)$ we have \cite[prop.\ 1]{li2018strong}:
\begin{align}
\EFI{\rvx}{\rvy} \geq \sum_y \int_0^1 \phi(P_\rvx(p(y \mid \rvx) \geq t)) \, dt - \MI{\rvx}{\rvy}
\label{eq:li_efi_bound}
\end{align}
By Bayes' rule, we have
\begin{align*}
p(y \mid \rvx = x) = p(y) \cdot \frac{dP_{\rvx \mid \rvy = y}}{dP_\rvx}(x)
\end{align*}
where $p(y)$ is the probability mass function of $P_\rvy$.
Now, letting ${w_y(h) = \Prob_{\rvx \sim P_\rvx}\left[\nicefrac{dP_{\rvx \mid \rvy = y}}{d P_{\rvx}}(\rvx) \geq h \right]}$, we get
\begin{align*}
\sum_y \int_0^1 \phi(P_\rvx(p(y &\mid \rvx) \geq t)) \, dt \\
&= \sum_y \int_0^1 \phi\left(w_y\left(\frac{t}{p(y)}\right)\right) \, dt \\
&= \sum_y p(y) \int_0^{\nicefrac{1}{p(y)}} \phi\left(w_y\left(h\right)\right) \, dh \\
&= \Exp_{\rvy \sim P_\rvy}\left[\CSD{P_{\rvx \mid \rvy}}{P_\rvx}\right].
\end{align*}
Hence, we can rewrite the bound in \cref{eq:li_efi_bound} as
\begin{align}
\EFI{\rvx}{\rvy} \geq \Exp_{\rvy \sim P_\rvy}\left[\CSD{P_{\rvx \mid \rvy}}{P_\rvx}\right] - \MI{\rvx}{\rvy}
\end{align}
\par
\textbf{Connection with the layered entropy:}
Recently, Hegazy and Li studied a subclass of one-shot channel simulation problems on $\Reals$, namely when $\rvx = \rvy + \rvz$ for $\rvy \perp \rvz$ with $\rvy \sim \Unif(0, t)$ for some $t > 0$ and $\rvz \sim P_\rvz$.
Assuming $\rvz$ has a density with respect to the Lebesgue measure $\lambda$ on $\Reals$ and setting $f = dP_\rvz/d\lambda$, they define its \textit{layered entropy} as
\begin{align}
h_L(Z) = \int_0^\infty \lambda(f \geq h) \logtwo \lambda(f \geq h) \, dh.
\end{align}
In fact, $h_L$ can be related to the channel simulation divergence via a limiting argument.
Namely, let $P_\rvz$ with density $f$ be fixed.
Let $P_{\rvz_n}$ with densities $f_n$ be a sequence of probability distributions, with $\supp f_n \subseteq (-n/2, n/2)$ and $f_n \to f$ pointwise as $n \to \infty$.
Then,
\begin{align}
h_L(Z) = \lim_{n \to \infty} \left(\logtwo n - \CSD*{P_{\rvz_n}}{\Unif\left(-\frac{n}{2}, \frac{n}{2}\right)}\right)
\label{eq:layered_entropy_limit}
\end{align}
To see this, first let us introduce the shorthand $U_n \sim \Unif(-n/2, n/2)$.
Then, observe that
\begin{align}
\frac{dP_{\rvz_n}}{d\lambda} = n\frac{dP_{\rvz_n}}{d P_{U_n}}.
\nonumber
\end{align}
Hence,
\begin{align}
\lambda(f_n \geq h) &= \lambda\left(\frac{dP_{\rvz_n}}{d P_{U_n}} \geq \frac{h}{n} \right) = n \cdot \Prob_{\rvu \sim U_n}\left[ \frac{dP_{\rvz_n}}{d P_{U_n}}(\rvu) \geq \frac{h}{n} \right].
\nonumber
\end{align}
Letting $w_n(h) = \Prob_{\rvu \sim U_n}\left[ \nicefrac{dP_{\rvz_n}}{d P_{U_n}}(\rvu) \geq h \right]$, we get
\begin{align}
h_L(Z_n) &= \int_0^\infty n\cdot w_n(n \cdot h) \logtwo (n \cdot w_n(n \cdot h)) \, dh
\nonumber \\
&= \logtwo n + \int_0^\infty w_n(\eta) \logtwo w_n(\eta) \, d\eta \nonumber\\
&= \logtwo n - \CSD{P_{\rvz_n}}{P_{U_n}}. \nonumber
\end{align}
Finally, taking the limit as $n \to \infty$, we obtain \cref{eq:layered_entropy_limit}.
In Theorems 1 and 3 of \cite{hegazy2022randomized}, Hegazy and Li show that the optimal rate $R^*_\infty(\rvz)$ to communicate $\rvz$ using dithered quantization, assuming a ``uniformly random source on $\Reals$'' is
\begin{align}
    R^*_\infty(\rvz) = -h_L(\rvz),
\end{align}
which can be thought of as an asymptotic analogue of our \Cref{thm:crs_index_entropy_lower_bound,thm:grs_achievability}.
\section{Convexity of Channel Simulation Divergence}
\label{appendix:csd_convexity}

We generalize the result that $\CSD{Q}{P}$ is convex with the following theorem.

\begin{theorem}
\label{thm:convexity}
    Let $\phi: [0, 1] \rightarrow \mathbb{R}$ be non-negative, concave with $\phi(0) = 0$. Let $(X, \mu)$ and $(Y, \nu)$ be finite measure spaces, with $\mu(X) = 1$. For each $y \in Y$ let $f_y(x)$ be a non-negative function and set $F(x) = \int_{Y} f_y(x) \,d\nu(y)\,$. Define:
    \begin{align*}
        w_y(h) = \mu(f_y(x) \ge h) \ \ \ \ \ \ \ W(h) = \mu(F(x) \ge h)
    \end{align*}
    Then we have:
    \begin{align}
        \label{ineq:concavity}
        \int_{Y} \int_{h = 0}^{\infty}{\phi(w_y(h)) \,dh} \,d\nu(y) \geq \int_{h = 0}^{\infty}{\phi(W(h)) \,dh}
    \end{align}
\end{theorem}

\begin{proof}
    Consider a sequence of piece-wise linear functions $\phi_n$ containing $2^n + 1$ nodes:
    \begin{align*}
        \phi_n\left(\frac{l}{2^n}\right) = \phi\left(\frac{l}{2^n}\right) \ \ \ \ \ l = 0, 1, ... , 2^n
    \end{align*}
    Since $\phi$ is concave, it is continuous on $(0, 1)$. We conclude that $\phi_n$ increases point-wise to the limit $\phi$ on whole $[0, 1]$. As $\phi_n$ are non-negative, by Monotone Convergence Theorem:
    \begin{align*}
        \int_{h = 0}^{\infty}{\phi(W(h)) \,dh} = \lim_{n \rightarrow \infty} \int_{h = 0}^{\infty}{\phi_n(W(h)) \,dh}
    \end{align*}
    with similar limits holding for $\phi(w_y(h))$. Therefore it suffices to show:
    \begin{align*}
        \int_{Y} \int_{h = 0}^{\infty}{\phi_n(w_y(h)) \,dh} \,d\nu(y) \geq \int_{h = 0}^{\infty}{\phi_n(W(h)) \,dh}
    \end{align*}
    Now we notice that each $\phi_n(x)$ can be written as a positive linear combination of non-negative functions of the form $\alpha x + \beta \min(x, c)$, (where $\beta \ge 0$) hence we are in the same situation as where we left off in the main text.
\end{proof}
We note:
\begin{itemize}
    \item[1.] It can be shown that any $\phi: [0, 1] \rightarrow \mathbb{R}$ with $\phi(0) = 0$ that satisfies \Cref{thm:convexity} has to be concave.
    \item[2.] The equality in \Cref{thm:convexity} occurs when functions $f_y(x)$ share the same super-level sets.
\end{itemize}
\vspace{0.2cm}
\begin{definition}[Generalization of $\CSD{Q}{P}$]
Fix a concave, continuous function $\phi: [0, 1] \to [0, \infty)$ such that $\phi(0) = \phi(1) = 0$.
For distributions $Q$ and $P$ with $Q \ll P$ and $w(h) = P\left(\nicefrac{dQ}{dP} \ge h\right)$, we define
\begin{align*}
\phiD{Q}{P} = \int_{0}^{\infty}{\phi(w(h)) \,dh}.
\end{align*}
\end{definition}

Since $\phi(0) = \phi(1) = 0$ and $\phi$ is concave, then $\phi \ge 0$.
Thus $\phiD{Q}{P}$ is always non-negative.
Furthermore, assuming that $\phi$ is non-zero, we find $\phi > 0$ on $(0, 1)$, thus $\phiD{Q}{P} = 0$ if and only if $w(h) \in \{0, 1\}$ almost everywhere.
This is equivalent to $P = Q$.

We further note, by \Cref{thm:convexity}, that $D^\phi$ is convex in $Q$ for fixed $P$. Note that $D^\phi$ may not be convex in both $P$ and $Q$ in general.
\par
Furthermore, if we assume that $\phi$ is twice differentiable on $(0, 1)$ and $\abs{\phi'(1)} < \infty$, 
then we can obtain the following integral representation of $\phiD{Q}{P}$ akin to \cref{eq:channel_simulation_divergence_integral_rep}:
\begin{align}
\phiD{Q}{P} = \phi'(1) - \int_0^1 \phi''(y) \int_0^\infty \min\{w(h), y\}\,dh \, dy.
\end{align}
We note, that this integral representation is similar in flavor to the integral representation of Sason and Verd{\'u}'s integral representation for $f$-divergences (Section VII in \cite{sason2016f}), with the quantity
\begin{align}
\mu_y(w) = \int_0^\infty \min\{w(h), y\} \, dh
\end{align}
playing an analogous role to the hockey-stick divergence.
To see this, let ${\mu_y' = \int_0^\infty \Ind[w(h) \geq y] \, dh}$ denote the derivative $\mu_y$ with respect to $y$ and observe that
\begin{align*}
\phiD{Q}{P} &= \int_0^\infty \phi(w(h))\,dh \\
&= -\int_0^1 \phi(y) \, d\mu_y'(w) \\
&= \int_0^1 \phi'(y)\mu'_y(w) \, dy \\
&= \phi'(1) - \int_0^1 \phi''(y) \mu_y(w) \, dy,
\end{align*}
where the last two equalities follow from integration by parts.
\par
For general $\phi$ twice differentiable we also have the following representation:
\begin{align*}
    \phi(y) = \int_{0}^{1} -\phi''(t) \min\{(1-y)t, y(1-t)\} \,dt 
\end{align*}
\section{Upper Bounds on (Alternative) Channel Simulation Divergence}
\label{appendix:alternative_csd}
\par
We define the Alternative Channel Simulation Divergence. Given $P, Q$ with $w(h) = P\left(\nicefrac{dQ}{dP} \ge h\right)$ we set:
\begin{align*}
    \ACSD{Q}{P} = \int_{0}^{\infty}{\Ent_b[w(h)] \,dh} 
\end{align*}
Thus $\KLD{Q}{P} \le \CSD{Q}{P} \le \ACSD{Q}{P}$.
Consider the following problem: For fixed upper bound on $\KLD{Q}{P}$, what is the maximal possible value that $\CSD{Q}{P}$ (alternatively $\ACSD{Q}{P}$) can take?

The solutions can be found with variational methods.
\begin{theorem}
\label{thm:opt_DCS}
Fix some number $\alpha \in (0, 1)$, set $p = \frac{1}{1 - \alpha}$ and suppose that for some $P_\alpha, Q_\alpha$:
\begin{align*}
w_\alpha(h) = \begin{cases}
1 & h \le \alpha
\\
(h / \alpha)^{-p} & h > \alpha
\end{cases}
\end{align*}
Then:
\begin{align*}
\KLD{Q_{\alpha}}{P_{\alpha}} &= \left(\frac{1}{\alpha} - 1 + \ln(\alpha)\right) / \ln(2)
\\
\CSD{Q_{\alpha}}{P_{\alpha}} &= \frac{1 - \alpha}{\alpha} / \ln(2)
\end{align*}
Furthermore, for any distributions $P, Q$:
\begin{align*}
    \KLD{Q}{P} \le \KLD{Q_\alpha}{P_\alpha} \ \Rightarrow \ \CSD{Q}{P} \le \CSD{Q_\alpha}{P_\alpha}
\end{align*}
\end{theorem}

With similar bounds holding for alternative channel simulation divergence:
\begin{theorem}
\label{thm:opt_ADCS}
Fix some number $\alpha > 1$, set $\beta = \frac{\pi}{\alpha} \sin\left(\frac{\pi}{\alpha}\right)^{-1}$ and suppose that for some $P_\alpha, Q_\alpha$:
\begin{align*}
w_\alpha(h) = \frac{1}{1 + (\beta h)^{\alpha}}
\end{align*}
Then:
\begin{align*}
\KLD{Q_{\alpha}}{P_{\alpha}} &= -\left(\ln(\beta) - 1 + \beta \cos\left(\frac{\pi}{\alpha}\right)\right) / \ln(2)
\\
\ACSD{Q_{\alpha}}{P_{\alpha}} &= \left(\alpha - \pi \cot\left(\frac{\pi}{\alpha}\right) \right) / \ln(2)
\end{align*}
Furthermore, for any distributions $P, Q$:
\begin{align*}
    \KLD{Q}{P} \le \KLD{Q_\alpha}{P_\alpha} \Rightarrow \ACSD{Q}{P} \le \ACSD{Q_\alpha}{P_\alpha}
\end{align*}
\end{theorem}

We note that all decreasing functions $w: [0, \infty) \to [0, 1]$ satisfying $\int w(h)\,dh = 1$ have a corresponding pair of distributions $P, Q$, thus the above bounds are achievable.

Now fix some $P, Q$ and set $\kappa = \KLD{Q}{P}$. Using \Cref{thm:opt_DCS} one can show:
\begin{align*}
    \CSD{Q}{P} \le \kappa + \logtwo(\kappa + 1) + \logtwo(\ln(4))
\end{align*}
Which in combination with \Cref{thm:grs_achievability} gives:
\begin{align*}
    \Ent[K] &\le \kappa + \logtwo(\kappa + 1) + \logtwo((e + 1) \ln(4))
    \\
    &< \kappa + \logtwo(\kappa + 1) + 2.366
\end{align*}
where $K$ is the index returned by GRS. 

Lastly, the equality $\CSD{Q}{P} = \KLD{Q}{P}$ holds when $w(h) = \frac{1}{c} \cdot 1[h \le c]$ for some $c \ge 1$.
\section{Extended Causal Rejection Samplers}
\label{appendix:causal_rejection_samplers}

The key results of our paper can be also proven for ECRS.

\begin{theorem}
\label{thm:entropy_lower_bound_for_ECRS}
    Let $N_K$ be the index returned by a Causal Rejection Sampler. Then:
    \begin{align*}
        \ACSD{Q}{P} \leq \Ent[N_K]
    \end{align*}
\end{theorem}

\begin{proof}
    One can prove a lemma similar to \Cref{lemma:index_entropy_csd_difference}, but for ECRS. Then, one can show that the thesis of \Cref{thm:entropy_lower_bound_for_ECRS} holds for any ECRS with uniformly bounded $K$. The idea behind the remaining part of the proof is to take "bounded" or "terminated" approximations of sequence $(N_1, \dots, N_K)$.

    Let $0$ be a symbol not belonging to $I$, and let $X_0 \sim P$ be an independent variable. For each natural number $L$ define a new ECRS in the following way: if $K \le L$, return sequence $(N_1, \dots, N_K)$, otherwise return sequence $(N_1, \dots, N_L, 0)$. One can check that this construction defines an ECRS.
    
    Let $N^L$ be the last index of the returned sequence. As in proof of \Cref{thm:crs_index_entropy_lower_bound}, one finds:
    \begin{align}
        \label{CRS_epic}
        \Ent[N^L] \rightarrow \Ent[N_K]
    \end{align}
    Let $Q_L$ be the distribution of $X_{N^L}$. Following similar reasoning to that in proof of \Cref{thm:entropy_lower_bound_for_ECRS}, it suffices to show:
    \begin{align}
        \label{CRS_CSD}
        \limsup_{L \rightarrow \infty} \CSD{Q_L}{P} \ge \CSD{Q}{P}
    \end{align}
    That will be our goal.

    To prove limit (\ref{CRS_CSD}) begin with the following observation:
    \begin{align*}
        \int \left|\frac{dQ_L}{dP} - \frac{dQ}{dP} \right| \,dP \to 0 \ \ \text{as} \ \ L \to \infty
    \end{align*}
    which is true since $K$ is always finite. Set $w(h) = P\left(\nicefrac{dQ}{dP} \ge h \right)$ and $w_L(h) = P\left(\nicefrac{dQ_L}{dP} \ge h \right)$. One can prove:
    \begin{align*}
        \int_{0}^{\infty}{\left|w_L(h) - w(h) \right| \,dh} \le \int \left|\frac{dQ^L}{dP} - \frac{dQ}{dP} \right| \,dP
    \end{align*}
    And thus we find:
    \begin{align}
        \label{CRS_limit}
        \int_{0}^{\infty}{\left|w_L(h) - w(h) \right| \,dh} \rightarrow 0 \ \ \text{as} \ \ L \rightarrow \infty
    \end{align}
    Next, let $M = \sup\{h: w(h) > 0 \}$ and let $\phi(x) = -x \logtwo(x)$. We note the following limit:
    \begin{align*}
        \int_{0}^{m}{\phi(w(h)) \,dh} \rightarrow \CSD{Q}{P} \ \ \text{as} \ \ m \rightarrow M
    \end{align*}
    Fix some $m < M$. Since $w(m) > 0$, it follows that there exists a constant $A$ satisfying $|\phi(w(h)) - \phi(w_L(h))| \le A |w(h) - w_L(h)|$ for all $0 \le h \le m$. Then:
    \begin{align*}
        \CSD{Q_L}{P} &= \int_{0}^{\infty}{\phi(w_L(h)) \,dh}
        \\
        &\ge \int_{0}^{m}{\phi(w(h)) - |\phi(w_L(h)) - \phi(w(h))| \,dh}
        \\
        &\ge \int_{0}^{m}{\phi(w(h)) \,dh} - A \int_{0}^{m}{|w_L(h) - w(h)| \,dh}
        \\
        &\rightarrow \int_{0}^{m}{\phi(w(h)) \,dh}
    \end{align*}
    By \cref{CRS_limit}. Taking $m \rightarrow M$ we get \cref{CRS_CSD}.
\end{proof}

It should be noted that the equality case $\ACSD{Q}{P} = \Ent[K]$ occurs "often" for Greedy Rejection Sampler. Indeed, if $w(h) = const$ holds Lebesgue a.e. in each interval $[L_k, L_{k + 1}]$, then we get equality in \Cref{thm:entropy_lower_bound_for_ECRS}.

Next we prove the following lemma.

\begin{lemma}
Consider a decreasing function $f: [0, 1] \to [m, M]$ where $0 \le m \le M \le 1$.
Set $I = \int_{0}^{1}{f(x) \,dx}$ and let $\phi(x) = -x \logtwo(x)$.
Then:
\begin{align*}
\Ent_b[I] - \logtwo(e + M) \le \int_{0}^{1}{\phi(f(x)) \,dx} - (1-I) \logtwo(e + m)
\end{align*}
\end{lemma}
\begin{proof}
Note that $m \le I \le M$. Define a decreasing function $g: [0, 1] \to [m, M]$ by:
\begin{align*}
g(x) =
\begin{cases}
M & x \le \frac{I - m}{M - m}
\\
m & \textit{otherwise}
\end{cases}
\end{align*}
One can verify that $\int_{0}^{1}g(x) \,dx = I = \int_{0}^{1}f(x) \,dx$. Furthermore, we find $\int_{0}^{a}g(x) \, dx \ge \int_{0}^{a}f(x) \,dx$ for each $a \in [0, 1]$. Thus $g$ majorizes $f$ in the sense defined in \cite{horváth2023uniform}. By Theorem 3 of \cite{horváth2023uniform} we thus obtain:
\begin{align*}
    \int_{0}^{1}{\phi(f(x)) \,dx} &\ge \int_{0}^{1}{\phi(g(x)) \,dx} 
    \\
    &= \frac{(I - m) \phi(M) + (M - I) \phi(m)}{M - m}
\end{align*}
Thus we strengthen the thesis:
\begin{align*}
&\: \Ent_b[I] - \logtwo(e + M) + (1-I)\logtwo(e + m)
\\
\le&\: \frac{(I - m) \phi(M) + (M - I) \phi(m)}{M - m}
\end{align*}
and strengthening more still (while changing basis of logarithm to natural all throughout):
\begin{align}
&\: \Ent_b[I] - I - \ln(1 + M/e) + \ln(1 + m/e)
\nonumber
\\
-&\: \frac{(I - m) \phi(M) + (M - I) \phi(m)}{M - m} \le 0
\label{ineq:stupid_lemma}
\end{align}

Fix some $0 \le m_0 \le I_0 \le M_0 \le 1$ and suppose that the expression (\ref{ineq:stupid_lemma}) attains a local maximum at $(m_0, I_0, M_0)$.
We consider 3 cases.

Firstly assume that $0 < m_0 \le I_0 < M_0 \le 1$.
We parameterize $M(t) = M_0 - t$, $I(t) = I_0$, $m(t) = m_0 - t$. 
Seeing how expression $(\ref{ineq:stupid_lemma})$ attains local maximum at $(m_0, I_0, M_0)$, the derivative of $(\ref{ineq:stupid_lemma})$ w.r.t. $t$ at $t = 0$ must be non-positive.
In other words:
\begin{align}
\label{eq:stupid2}
\frac{M_0 + m_0}{M_0 - m_0} \ln(M_0/m_0) \le 1 - \frac{1}{M_0 + e} + \frac{1}{m_0 + e}
\end{align}
However, one can verify that for all $m_0, M_0 > 0$:
\begin{align*}
    \frac{M_0 + m_0}{M_0 - m_0} \ln(M_0/m_0) \ge 2 > 1 - \frac{1}{M_0 + e} + \frac{1}{m_0 + e}
\end{align*}
Thus (\ref{eq:stupid2}) is not satisfied - a contradiction.

The second case is $I_0 = M_0$. Then one can verify that (\ref{ineq:stupid_lemma}) holds whenever $I = M$ or $I = m$. Hence for the rest of the proof we will assume that $I_0$ satisfies $m_0 < I_0 < M_0$.

The last case is $m_0 = 0$. Then ineq. (\ref{ineq:stupid_lemma}) takes form:
\begin{align}
\Ent_b[I_0] - I_0 - \ln(1 + M_0/e) + I_0 \ln(M_0) \le 0
\label{ineq:stupid_lemma0}
\end{align}
The expression (\ref{ineq:stupid_lemma0}) is concave in $I_0$.
Differentiating w.r.t $I$, while noting that $m_0 < I_0 < M_0$ is a local maximum, we deduce that $I_0 = \nicefrac{M_0}{(e + M_0)}$.
Substituting into (\ref{ineq:stupid_lemma0}) we find an equality. Hence (\ref{ineq:stupid_lemma}) holds for all $0 \le m \le I \le M \le 1$, which finishes the proof.
\end{proof}
\section{Additional Results for the Greedy Rejection Sampler}
\label{appendix:grs_additional_results}

The constant term in inequality of \Cref{thm:grs_achievability} is optimal.
Indeed, let:
\begin{align*}
w_\epsilon(x) = \begin{cases}
1 & x \le \frac{1}{1+e}
\\
\epsilon & \frac{1}{1+e} < x \le \frac{1}{1+e} + \frac{e}{(1+e)\epsilon}
\\
0 & \text{otherwise}
\end{cases}
\end{align*}
Then $\Ent[K] - \CSD{Q}{P} \to \logtwo(e + 1)$ as $\epsilon \to 0$, where $K$ is index returned by GRS.

Furthermore, we can extend \Cref{thm:grs_achievability} to alternative channel simulation divergence. The proof is the same, but it relies on a lemma:
\begin{lemma}
\label{lemma:ACSD_GRS_bound}
Fix $P, Q$ and $k \ge 1$. Let $K$ be the index returned by GRS. Then:
\begin{align*}
\begin{split}
\Ent[K] \le &f(1) + \int_{0}^{L_k}{\Ent_b(w(h)) \, dh}
\\
&+ S_k \left(\Ent[K | K \ge k] - f(w(L_k))\right)
\end{split}
\end{align*}
where $f(x) = \logtwo(1 + x (1-x)^{\nicefrac{(1-x)}{x}}) = \logtwo(1 + e^{-\Ent_b[x]/x})$.
\end{lemma}
Similarly to \Cref{thm:grs_achievability}, the proof of \Cref{lemma:ACSD_GRS_bound} also relies on induction. Since $f(1) = 1$ we obtain a new upper bound.
\begin{theorem}
\label{thm:bound_by_ACSD}
    Let $K$ be the index returned by GRS. Then:
    \begin{align*}
        \Ent[K] \le \ACSD{Q}{P} + 1
    \end{align*}
\end{theorem}
As before, the constant in \Cref{thm:bound_by_ACSD} is optimal.




\end{document}